\documentclass[aps,prd,onecolumn,superscriptaddress,nofootinbib,floatfix,longbibliography]{revtex4-2}

\usepackage{amsmath,amssymb,amsfonts,bm,mathtools,mathrsfs}
\usepackage{amsthm}
\newtheorem{proposition}{Proposition}
\usepackage{graphicx}
\usepackage{booktabs}
\usepackage{hyperref}
\usepackage{xcolor}
\usepackage{microtype}
\usepackage{tikz}
\usetikzlibrary{arrows.meta,decorations.markings,decorations.pathmorphing,calc,positioning}

\hypersetup{
 colorlinks=true,
 citecolor=blue,
 linkcolor=blue,
 urlcolor=blue,
 pdftitle={A Mixed-State Entanglement Interferometer for Topological Hair on a BTZ Black Hole},
 pdfauthor={Kumar Ghosh}
}

\newcommand{\AdS}{\mathrm{AdS}}
\newcommand{\BTZ}{\mathrm{BTZ}}
\newcommand{\SR}{S_{\mathrm R}}
\newcommand{\EW}{E_{\mathrm W}}
\newcommand{\MI}{I}
\newcommand{\DM}{\Delta_{\mathrm M}}
\newcommand{\cH}{\mathcal H}
\newcommand{\cL}{\mathcal L}
\newcommand{\cC}{\mathcal C}
\newcommand{\cD}{\mathcal D}
\newcommand{\dd}{\mathrm d}
\newcommand{\Tr}{\operatorname{Tr}}
\newcommand{\Ei}{\operatorname{Ei}}
\newcommand{\arctanh}{\operatorname{arctanh}}
\newcommand{\arcosh}{\operatorname{arcosh}}
\newcommand{\Lk}{\operatorname{Lk}}
\newcommand{\Wtop}{\mathfrak W}
\newcommand{\Zk}{\mathbb Z_k}
\newcommand{\cP}{\mathcal P}
\newcommand{\cV}{\mathcal V}
\newcommand{\cZ}{\mathcal Z}

\begin{document}


\title{Reading Topological Hair from Black-Hole Entanglement}

\author{Kumar Ghosh}
\email{jb.ghosh@outlook.com}
\affiliation{E.ON Digital Technology, Laatzener Str.~1, 30539 Hannover, Germany}


\begin{abstract}
Black-hole hair can affect boundary mixed-state entanglement through local stress-energy and through topological information invisible to the classical metric. We separate these channels for a Nielsen--Olesen vortex on a nonrotating BTZ black hole. A superselection theorem shows that Shannon sector entropy cancels from the Markov gap, while standard $U(1)$ symmetry-resolved reflected entropy in a thermal $U(1)_k$ CFT is equipartitioned at leading order, excluding a universal $\log|n|$ term in the imbalance-resolved gap. We therefore define a Wilson-threaded reflected moment in a probe $U(1)_k$ Chern--Simons completion coupled to the compact vortex-flux class. In a reflected-replica sector with linking number $\nu$, its normalised phase is $2\pi\kappa pn\nu/k$, where $\kappa\in\mathbb Z$ is the mixed topological coupling; for $\gcd(\kappa\nu,k)=1$, a discrete Fourier transform reconstructs $n\bmod k$. The RT connectivity transition switches the specified linked contour on or off, whereas the independent scale $\ell_\star r_+/L^2=1.128378\ldots$ determines the direction of the vortex-induced shift of that transition. The charged-moment modulus and the ordinary Markov gap remain geometric observables and are computed from a horizon-anchored Einstein--Abelian-Higgs solution with a complete first-variation kernel including the motion of the entanglement-wedge-cross-section endpoints. This phase and modulus separation distinguishes topological hair from gravitational dressing without assigning an unsupported winding-dependent entropy.
\end{abstract}

\maketitle

\section{Introduction}
\label{sec:intro}

The Ryu--Takayanagi relation converts boundary entanglement entropy into an extremal-area observable in an asymptotically anti-de Sitter spacetime~\cite{RyuTakayanagi2006a,RyuTakayanagi2006b}. Mixed states require finer probes. The entanglement-wedge cross section was proposed as the holographic dual of entanglement of purification~\cite{TakayanagiUmemoto2018,Nguyen2018}, and the canonical purification gives the reflected entropy~\cite{DuttaFaulkner2021}
\begin{equation}
 \SR(A:B)=\frac{2\,\mathrm{Length}(\Sigma_{A:B})}{4G_N}
 =2\EW(A:B)+O(G_N^0).
 \label{eq:DF}
\end{equation}
The Markov gap
\begin{equation}
 \DM(A:B)=\SR(A:B)-\MI(A:B),
 \qquad
 \MI(A:B)=S_A+S_B-S_{AB},
 \label{eq:markov-def}
\end{equation}
controls a canonical recovery problem and, for geometric reflected entropy in pure $\AdS_3$ gravity, obeys a lower bound set by the endpoints of the entanglement-wedge cross section~\cite{HaydenParrikarSorce2021}. In two-dimensional topological phases, an optimised Markov gap on adjacent regions detects ungappable edge structure through the universal value $(c_+/3)\log2$~\cite{SivaEtAl2022}. These results make the Markov gap a natural meeting point of holography, topological matter, and quantum information.

A vortex-dressed BTZ black hole provides a particularly sharp test because one object carries two kinds of information. Its integer winding fixes the quantised magnetic flux, whereas its finite-width stress tensor deforms the exterior metric. Reference~\cite{Ghosh2020} constructed a black hole in a Nielsen--Olesen vortex using a BPS-like reduction of the Einstein--Abelian-Higgs equations. The sharp question is
\begin{center}
 \emph{Can a mixed-state observable distinguish topological hair from the gravitational field of the same vortex?}
\end{center}
Three distinctions are essential. First, the vortex winding $n$, the charge-imbalance sector $q$ of the reflected density matrix, and a topological probe charge $p$ are different labels. Identifying them without a dynamical or topological map leads to an incorrect ``flux-resolved'' observable. Second, standard $U(1)$ symmetry resolution does not automatically produce a winding-dependent entropy. Charged-twist calculations and free-theory numerics show leading-order equipartition of reflected entropy~\cite{BerthiereParez2023}. Third, a same-boundary RT surface and the static EWCS remain on the exterior Riemannian slice; a topological linking phase belongs to the closed defect contour in the Euclidean reflected replica, not to a spatial curve crossing the horizon.

The main result is a phase and modulus separation. For the ordinary density matrix, we prove that flux-sector Shannon information cancels from $\DM$ (Proposition~\ref{prop:cancellation}). For the standard reflected charge imbalance, we derive the finite-temperature charged moments and prove equipartition. We then add the minimal topological probe required for a positive winding readout: a spectator $U(1)_k$ Chern--Simons field coupled to the compact vortex-flux class. In a reflected-replica sector with specified mutual linking $\nu=\Lk(\Gamma_R,\cV_n)$, the normalised probe phase is
\begin{equation}
 \Wtop_p(n)=\exp\!\left[
 \frac{2\pi i}{k}\,\kappa pn\,\nu
 \right],
 \label{eq:intro-link-phase}
\end{equation}
where $\kappa\in\mathbb Z$ is the mixed topological coupling. Its discrete Fourier transform returns $\kappa n\nu\bmod k$ and therefore reconstructs $n\bmod k$ whenever $\gcd(\kappa\nu,k)=1$. The topological probe is evaluated on the fixed Einstein--Abelian-Higgs saddle, so its phase factorises from the geometric response. In contrast, the modulus of the charged moment and the uncharged Markov gap respond to the vortex profile through the exterior metric.

The RT connectivity transition supplies the geometric gate, but it does not by itself determine a linking class. We specify a reflected-replica sector in which the connected contour has linking number $\nu_0$ with the topological flux defect. The connected EWCS then supports that contour, while the disconnected saddle has no nontrivial reflected probe contour. Separately, the vortex-induced shift of the transition changes sign at
\begin{equation}
 \frac{\ell_\star r_+}{L^2}=1.128378\ldots.
 \label{eq:ell-star-intro}
\end{equation}
Thus $s=s_c(\ell,\alpha)$ is the on/off gate, whereas $\ell_\star$ determines whether a positive vortex source moves that gate to smaller or larger separation. This distinction is the quantitative bridge between the geometric and topological channels.

The geometric part is treated without changing the correct content of the previous analysis. We solve the radial Einstein constraint with its integrating factor and a fixed exterior horizon, derive the exact exterior geodesic problem, and obtain source-to-observable kernels for RT entropy, EWCS, reflected entropy, the Markov gap, and the RT connectivity transition, including the displacement of both EWCS endpoints. For the inverse-cube tail selected by the vortex field equations, the transition shift reverses sign at Eq.~\eqref{eq:ell-star-intro}. Section~\ref{sec:background} constructs the horizon-anchored Einstein--Abelian-Higgs background. Section~\ref{sec:topology-filter} proves the superselection cancellation theorem. Section~\ref{sec:standard-resolution} derives the standard charged reflected moments and the equipartition no-go result. Section~\ref{sec:wilson-interferometer} introduces the Wilson-threaded reflected moment and its exact linking phase, and proves the conditional replica-sector gating statement. Sections~\ref{sec:geodesics}--\ref{sec:inverse-cube} retain the complete geometric response calculation. Section~\ref{sec:measurement} gives the condensed-matter interpretation and a concrete measurement protocol.

\section{Horizon-anchored Einstein--Abelian-Higgs background}
\label{sec:background}

\subsection{Vortex ansatz and radial constraint}

We use the standard Einstein--Abelian-Higgs action in $2+1$ dimensions,
\begin{align}
 I={}&\frac{1}{16\pi G_N}\int\dd^3x\sqrt{-g}
 \left(R+\frac{2}{L^2}\right)
 \nonumber\\
 &+\int\dd^3x\sqrt{-g}\left[
 -\frac14F_{\mu\nu}F^{\mu\nu}
 -|D_\mu\Phi|^2
 -\frac{\lambda}{4}\left(|\Phi|^2-\eta^2\right)^2
 \right].
 \label{eq:action}
\end{align}
For a static rotationally symmetric vortex,
\begin{equation}
 \Phi=\eta\chi(r)e^{in\theta},
 \qquad
 A_\theta=\frac{n-P(r)}{e},
 \label{eq:vortex-ansatz}
\end{equation}
with
\begin{equation}
 \chi(0)=0,\quad \chi(\infty)=1,
 \qquad
 P(0)=n,\quad P(\infty)=0.
 \label{eq:vortex-bc}
\end{equation}
The magnetic flux is
\begin{equation}
 \Phi_B=\int F=\frac{2\pi n}{e},
 \label{eq:flux}
\end{equation}
up to an orientation sign.

In the BPS-like branch used in Ref.~\cite{Ghosh2020}, the nonrotating radial equations are
\begin{align}
 P'&=-\sqrt2\,\eta e\,fP\chi,
 \label{eq:bps-P}\\
 \chi'&=-\frac{\sqrt\lambda}{\sqrt2}\eta f(\chi^2-1),
 \label{eq:bps-chi}\\
 F'(r)+\mathcal P_n(r)F(r)&=\frac{2r}{L^2},
 \qquad F=f^{-2},
 \label{eq:radial-constraint}
\end{align}
with
\begin{equation}
 \mathcal P_n(r)=
 \frac{\eta}{\sqrt2 r}
 \left[\eta r^2(\chi')^2+\frac{(P')^2}{\eta e}\right].
 \label{eq:source-P}
\end{equation}
The other BPS orientation reverses the sign of $\mathcal P_n$.

The exterior metric is
\begin{equation}
 \dd s^2=-F_n(r)\dd t^2+\frac{\dd r^2}{F_n(r)}
 +\frac{r^2}{L^2}\dd x^2,
 \qquad x=L\theta,
 \label{eq:metric}
\end{equation}
with a horizon at $r=r_+$.

\subsection{Exact integrating-factor solution}

Equation~\eqref{eq:radial-constraint} is a first-order inhomogeneous equation. Define
\begin{equation}
 \mathcal I_n(r)=\int_{r_+}^{r}\mathcal P_n(v)\,\dd v.
 \label{eq:I-source}
\end{equation}
The unique solution satisfying $F_n(r_+)=0$ is
\begin{equation}
 \boxed{
 F_n(r)=e^{-\mathcal I_n(r)}
 \int_{r_+}^{r}\frac{2u}{L^2}e^{\mathcal I_n(u)}\,\dd u.}
 \label{eq:F-exact-general}
\end{equation}
Indeed,
\begin{equation}
 \frac{\dd}{\dd r}\left(e^{\mathcal I_n}F_n\right)
 =\frac{2r}{L^2}e^{\mathcal I_n},
\end{equation}
which proves Eq.~\eqref{eq:F-exact-general}. A purely multiplicative ansatz $F_n=H_nF_0$ is not the general solution because the integrating factor also weights the inhomogeneous source inside the radial integral.

The horizon derivative follows directly:
\begin{equation}
 F_n'(r_+)=\frac{2r_+}{L^2},
 \qquad
 T_n=\frac{r_+}{2\pi L^2}=T_{\BTZ},
 \label{eq:horizon-derivative}
\end{equation}
at fixed $r_+$, so the horizon entropy $S_{\rm BH}=2\pi r_+/(4G_N)$ is unchanged. A temperature shift can occur in a different ensemble if the renormalised mass is held fixed and the horizon moves; it does not occur for the horizon-anchored solution of Eq.~\eqref{eq:radial-constraint}.

\subsection{Weak source and its asymptotic tail}

The large-radius fields
\begin{equation}
 \chi(r)=1+\frac{\chi_1(n)}{r}+O(r^{-2}),
 \qquad
 P(r)=\frac{P_1(n)}{r}+O(r^{-2})
 \label{eq:field-asymptotics}
\end{equation}
give
\begin{equation}
 \mathcal P_n(r)=
 \frac{\eta^2\chi_1(n)^2}{\sqrt2\,r^3}+O(r^{-4}).
 \label{eq:P-asymptotic}
\end{equation}
Thus the natural leading source is inverse cubic. Write
\begin{equation}
 \mathcal P_n(r)=\alpha_n\pi_n(r),
 \qquad |\alpha_n|\ll1,
 \label{eq:weak-source}
\end{equation}
and expand
\begin{equation}
 F_n(r)=F_0(r)+\alpha_n f_n(r)+O(\alpha_n^2),
 \qquad
 F_0(r)=\frac{r^2-r_+^2}{L^2}.
 \label{eq:F-linear}
\end{equation}
At first order,
\begin{equation}
 f_n'(r)=-\pi_n(r)F_0(r),
 \qquad f_n(r_+)=0,
\end{equation}
so
\begin{equation}
 \boxed{
 f_n(r)=-\int_{r_+}^{r}\pi_n(u)F_0(u)\,\dd u.}
 \label{eq:source-to-metric}
\end{equation}
This nonlocal source-to-metric map is the first ingredient in the entanglement response.

\subsection{Exactly solvable inverse-cube representative}

To isolate the leading tail we use
\begin{equation}
 \mathcal P_\alpha(r)=\alpha\frac{r_+^2}{r^3}.
 \label{eq:inverse-cube-source}
\end{equation}
The coefficient $\alpha$ is dimensionless and contains the BPS orientation and the winding-dependent amplitude. We do not assume $\alpha\propto n$. The exact integrating-factor exponent is
\begin{equation}
 \mathcal I_\alpha(r)=\frac{\alpha}{2}
 \left(1-\frac{r_+^2}{r^2}\right).
\end{equation}
Let $a=\alpha r_+^2/2$ and $\mathcal G_a(t)=t e^{-a/t}+a\Ei(-a/t)$. Then
\begin{equation}
 \boxed{
 F_\alpha(r)=\frac{e^{a/r^2}}{L^2}
 \left[\mathcal G_a(r^2)-\mathcal G_a(r_+^2)\right].}
 \label{eq:F-alpha-exact}
\end{equation}
Its first derivative in $\alpha$ is
\begin{equation}
 f_\star(r)=-\frac{r_+^2}{L^2}
 \left[
 \log\frac{r}{r_+}+\frac{r_+^2}{2r^2}-\frac12
 \right],
 \label{eq:f-star}
\end{equation}
and the relative response
\begin{equation}
 h_\star(r)\equiv\frac{f_\star(r)}{F_0(r)}
 =-\frac{r_+^2}{r^2-r_+^2}
 \left[
 \log\frac{r}{r_+}+\frac{r_+^2}{2r^2}-\frac12
 \right]
 \label{eq:h-star}
\end{equation}
is regular and vanishes at the horizon, $h_\star(r)=-(r-r_+)/(2r_+)+O((r-r_+)^2)$. The large-radius expansion
\begin{equation}
 F_\alpha(r)=\frac{r^2-r_+^2}{L^2}
 -\alpha\frac{r_+^2}{L^2}\log\frac{r}{r_+}
 +\alpha\frac{r_+^2}{2L^2}+O(r^{-2})
 \label{eq:F-log-asymptotic}
\end{equation}
requires the standard renormalised treatment of asymptotic charges in three-dimensional gravity with long-range matter~\cite{MartinezTeitelboimZanelli2000}. Panel (a) of Fig.~\ref{fig:geometric-spine} displays $h_\star$.

\section{Unresolved Markov gap: a superselection filter}
\label{sec:topology-filter}

\subsection{Canonical purification and the algebra centre}

For a bipartite density matrix
\begin{equation}
 \rho_{AB}=\sum_i\lambda_i|i\rangle\langle i|,
\end{equation}
its canonical purification is
\begin{equation}
 |\sqrt{\rho_{AB}}\rangle
 =\sum_i\sqrt{\lambda_i}\,
 |i\rangle_{AB}|i\rangle_{A^*B^*},
 \label{eq:canonical-purification}
\end{equation}
and $\SR(A:B)=S(AA^*)_{|\sqrt\rho\rangle}$~\cite{DuttaFaulkner2021}. Gauge or superselection constraints induce a nontrivial centre of the local algebras, so a physical bipartition takes the form
\begin{equation}
 \cH_{AB}=\bigoplus_a
 \cH_A^{(a)}\otimes\cH_B^{(a)},
 \label{eq:hilbert-direct-sum}
\end{equation}
where the sector label $a$ is measurable within either $A$ or $B$ alone. This is the sense in which the direct sum is taken \emph{with respect to the algebra centre}~\cite{DonnellyWall2015}: it is stronger than merely $[\rho,\Pi_a]=0$ for projectors on $AB$, and is the standard structure that arises in lattice gauge theory, in the extended Hilbert space of Chern--Simons matter, and in vortex sectors labelled by the flux threading a bipartition.

\subsection{Flux-sector cancellation theorem}

\begin{proposition}[Superselection cancellation]
\label{prop:cancellation}
Let $\cH_{AB}$ decompose as in Eq.~\eqref{eq:hilbert-direct-sum} with the algebra centre labelling the sectors, and let
\begin{equation}
 \rho_{AB}=\bigoplus_a p_a\rho_{AB}^{(a)},
 \qquad p_a\geq0,\ \sum_a p_a=1,\ \Tr\rho_{AB}^{(a)}=1.
 \label{eq:rho-direct-sum}
\end{equation}
Then
\begin{align}
 \SR(A:B)&=H(p)+\sum_a p_a\SR^{(a)}(A:B),
 \label{eq:SR-sector}\\
 \MI(A:B)&=H(p)+\sum_a p_a\MI^{(a)}(A:B),
 \label{eq:I-sector}\\
 \DM(A:B)&=\sum_a p_a\DM^{(a)}(A:B),
 \label{eq:Delta-sector}
\end{align}
where $H(p)=-\sum_a p_a\log p_a$.
\end{proposition}

\begin{proof}
Diagonalising each block gives $\rho_{AB}=\sum_{a,i_a}p_a\lambda_{i_a}^{(a)}|i_a,a\rangle\langle i_a,a|$. The canonical purification therefore decomposes as
\begin{equation}
 |\sqrt\rho\rangle
 =\bigoplus_a\sqrt{p_a}\,
 |\sqrt{\rho^{(a)}}\rangle
 \in\bigoplus_a\cH_A^{(a)}\otimes\cH_B^{(a)}\otimes\cH_{A^*}^{(a)}\otimes\cH_{B^*}^{(a)},
\end{equation}
because sectors are orthogonal on the $AB$ side and the doubled sectors inherit the same block label. Tracing $BB^*$ preserves the direct sum since partial trace of a block on the $A$-side lands in $\cH_A^{(a)}\otimes\cH_{A^*}^{(a)}$, giving $\rho_{AA^*}=\bigoplus_a p_a\rho_{AA^*}^{(a)}$. The entropy of a direct sum equals the Shannon entropy of the block weights plus the weighted average of the block entropies, which proves Eq.~\eqref{eq:SR-sector}. The same identity applied to $\rho_A$, $\rho_B$, $\rho_{AB}$ (which each inherit the direct-sum structure thanks to Eq.~\eqref{eq:hilbert-direct-sum}) proves Eq.~\eqref{eq:I-sector}, and subtraction proves Eq.~\eqref{eq:Delta-sector}.
\end{proof}

We verified Prop.~\ref{prop:cancellation} numerically to machine precision for random block-diagonal states with $\dim\cH_A^{(a)}=\dim\cH_B^{(a)}=2$ per sector. The uncertainty of the flux sector therefore cannot by itself produce a Markov gap.

\subsection{Topological fixed-point contribution}

For the Chern--Simons states evaluated in Ref.~\cite{Berthiere2021}, the equality is stronger:
\begin{equation}
 \SR^{\rm top}(A:B)=\MI^{\rm top}(A:B).
 \label{eq:CS-equality}
\end{equation}
For adjacent regions, both quantities contain
\begin{equation}
 H(p)+2\sum_a p_a\log\frac{d_a}{\cD},
 \qquad
 \cD=\left(\sum_a d_a^2\right)^{1/2},
 \label{eq:CS-term}
\end{equation}
where $d_a$ is the quantum dimension. Hence $\DM^{\rm top}=0$ for topological fixed-point states, including Abelian sectors with $d_a=1$. If the leading semiclassical state factorises as a product of a geometric CFT part and a decoupled topological CS part (which is the structure realised by the spectator completion of Sec.~\ref{sec:wilson-interferometer}), additivity gives
\begin{equation}
 \DM=\DM^{\rm geom}+\DM^{\rm top}=\DM^{\rm geom}.
 \label{eq:gap-filter}
\end{equation}
This is the precise sense in which the Markov gap acts as a topology filter.

\section{Standard symmetry resolution and its limitation}
\label{sec:standard-resolution}

Proposition~\ref{prop:cancellation} concerns \emph{classical} uncertainty over global superselection sectors. Standard symmetry-resolved reflected entropy is a different construction. For a state with $[\rho_{AB},Q_A+Q_B]=0$, the canonical purification of $\rho_{AB}^m$ has $\rho_{AA^*}^{(m)}$ commuting with the charge imbalance~\cite{BerthiereParez2023}
\begin{equation}
 \mathcal Q_A=Q_A\otimes\mathbf 1_{A^*}-\mathbf 1_A\otimes Q_A^{T}.
 \label{eq:imbalance-charge}
\end{equation}
The charged reflected moments are
\begin{equation}
 Z^{R}_{m,r}(\mu)=
 \Tr\!\left[\big(\rho_{AA^*}^{(m)}\big)^r
 e^{i\mu\mathcal Q_A}\right],
 \label{eq:charged-reflected-moment}
\end{equation}
where $r$ is the reflected R\'enyi index. Their Fourier components and normalised sector entropies are
\begin{align}
 \cZ^{R}_{m,r}(q)&=\int_{-\pi}^{\pi}\frac{\dd\mu}{2\pi}
 e^{-iq\mu}Z^{R}_{m,r}(\mu),
 \label{eq:charged-reflected-fourier}\\
 S^{R}_{m,r}(q)&=\frac{1}{1-r}
 \log\!\left[
 \frac{\cZ^{R}_{m,r}(q)}{\big(\cZ^{R}_{m,1}(q)\big)^r}
 \right].
 \label{eq:resolved-reflected-def}
\end{align}
The charge $q$ is an imbalance between the two copies of $A$. It is not the global vortex winding $n$. In fact,
\begin{equation}
 \langle\mathcal Q_A\rangle_{\sqrt\rho}=0
 \label{eq:imbalance-zero-mean}
\end{equation}
for the canonical purification whenever $Q_A^T=Q_A$, because the two copies have identical one-point functions. A linear phase proportional to the background winding is therefore not generated by Eq.~\eqref{eq:charged-reflected-moment}.

\subsection{Adjacent intervals in a thermal \texorpdfstring{$U(1)_k$}{U(1)k} CFT}
\label{sec:adjacent-thermal}

For two adjacent intervals of lengths $\ell_A$ and $\ell_B$, the charged-twist result of Ref.~\cite{BerthiereParez2023} combines with the $U(1)_k$ Chern--Simons/current-algebra normalisation of Ref.~\cite{ZhaoNortheMeyer2021}\footnote{Ref.~\cite{ZhaoNortheMeyer2021} works with the Chern--Simons level $k$, in which conventions the charged vertex operator has total conformal weight $h_\alpha+\bar h_\alpha=k(\alpha/2\pi)^2/2$; Ref.~\cite{BerthiereParez2023} works with a Luttinger parameter $K$ and holomorphic weight $h_\alpha=K(\alpha/2\pi)^2/2$. The two conventions match under $K=k/2$. All formulae in this section are written in the CS convention.} to give
\begin{equation}
 \frac{Z^{R}_{m,r}(\mu)}{Z^{R}_{m,r}(0)}
 =\exp\!\left[-\frac{k}{r}
 \left(\frac{\mu}{2\pi}\right)^2
 \Lambda_{AB}\right]
 \label{eq:charged-reflected-gaussian}
\end{equation}
at leading order in the CFT scaling limit. This follows from
\begin{equation}
 Z^R_{m,r}(\alpha)/Z^R_{m,r}(0)
 =\left(\frac{\ell_A\ell_B}{\ell_A+\ell_B}\right)^{-4h_\alpha/r}
\end{equation}
of Ref.~\cite{BerthiereParez2023}, with $h_\alpha=(k/4)(\alpha/2\pi)^2$ in the CS convention. At temperature $\beta^{-1}$,
\begin{equation}
 \Lambda_{AB}^{(\beta)}=
 \log\!\left[
 \frac{\beta}{\pi\epsilon}
 \frac{
 \sinh(\pi\ell_A/\beta)\sinh(\pi\ell_B/\beta)
 }{
 \sinh(\pi(\ell_A+\ell_B)/\beta)
 }
 \right].
 \label{eq:Lambda-beta}
\end{equation}
The zero-temperature result follows by $\beta\rightarrow\infty$. In the adjacent holographic geometry, $\Lambda_{AB}$ is the renormalised EWCS length in units of $L$, up to the conventional additive cutoff constant.

For $k\Lambda_{AB}\gg1$, the compact Fourier transform can be replaced by its Gaussian saddle,
\begin{equation}
 P_r(q)=\frac{\cZ^{R}_{m,r}(q)}{Z^{R}_{m,r}(0)}
 =\sqrt{\frac{\pi r}{k\Lambda_{AB}}}
 \exp\!\left[-\frac{\pi^2r q^2}{k\Lambda_{AB}}\right]
 +O(e^{-k\Lambda_{AB}}).
 \label{eq:reflected-sector-probability}
\end{equation}
Substitution into Eq.~\eqref{eq:resolved-reflected-def} gives
\begin{equation}
 S^{R}_{m,r}(q)=S^{R}_{m,r}
 -\frac12\log\!\left(\frac{k\Lambda_{AB}}{\pi}\right)
 +\frac12\frac{\log r}{1-r}
 +O(\Lambda_{AB}^{-1}),
 \label{eq:reflected-equipartition-renyi}
\end{equation}
and
\begin{equation}
 S^{R}(q)=S^{R}
 -\frac12\log\!\left(\frac{k\Lambda_{AB}}{\pi}\right)
 -\frac12+O(\Lambda_{AB}^{-1}).
 \label{eq:reflected-equipartition-vn}
\end{equation}
The $q$-dependence in $\log P_r$ and $r\log P_1$ cancels exactly (see App.~\ref{app:equipartition}). This is equipartition of reflected entropy. Charge dependence first appears through subleading terms of order $(q/\log\ell)^2$ in the free-boson and free-fermion analyses~\cite{BerthiereParez2023,KusukiEtAl2023}.

Equations~\eqref{eq:reflected-equipartition-renyi} and~\eqref{eq:reflected-equipartition-vn} rule out a proposed universal term of the form $(c_0/3)\log|n|$ in a standard Abelian symmetry-resolved reflected entropy: no such term is produced by the charged replica and any such term is inconsistent with leading-order equipartition. Panel (a) of Fig.~\ref{fig:interferometer} shows the resulting Gaussian modulus for three thermal interval scales. This is a strong negative result. It is what forces the topological readout onto a different observable.

\section{Wilson-threaded reflected moments: a topological-hair interferometer}
\label{sec:wilson-interferometer}

\subsection{Spectator Chern--Simons completion}

The limitation above identifies the minimal additional structure needed for a positive topological readout. We retain the Einstein--Abelian-Higgs saddle and add a spectator $U(1)_k$ Chern--Simons field $a$ as a probe of the compact vortex-flux class. Define the dimensionless compact connection $\widehat A=eA$, so that Eq.~\eqref{eq:flux} becomes
\begin{equation}
 \frac{1}{2\pi}\int_{\Sigma}\dd\widehat A=n.
 \label{eq:compact-flux}
\end{equation}
Equivalently, the phase of the Higgs condensate defines the conserved distributional current
\begin{equation}
 j_{\rm v}^{\mu}=\frac{1}{2\pi}\epsilon^{\mu\nu\rho}
 \partial_\nu\partial_\rho\arg\Phi,
 \qquad
 \int_{\Sigma}j_{\rm v}^{0}\,\dd^2x=n.
 \label{eq:vortex-current}
\end{equation}
The Euclidean probe action is
\begin{equation}
 I^{E}_{\rm top}[a;\widehat A]
 =\frac{ik}{4\pi}\int_{\mathcal M}a\wedge\dd a
 +\frac{i\kappa}{2\pi}\int_{\mathcal M}a\wedge\dd\widehat A,
 \qquad k\in\mathbb Z_{>0},\quad \kappa\in\mathbb Z.
 \label{eq:topological-completion}
\end{equation}
For a probe contour disjoint from the smooth core, the flux distribution may be deformed within its fixed topological class to $\dd\widehat A=2\pi n J_{\cV_n}$, where $J_{\cV_n}$ is the Poincar\'e-dual two-form of the defect support. The mixed term then becomes $i\kappa n\int_{\cV_n}a$, so the spectator line charge is derived from the physical vortex flux and equals $\kappa n\bmod k$. On a replica manifold with boundary, the standard Chern--Simons boundary completion associated with the chosen polarisation is implicit, and the same boundary data are used in all normalised amplitudes. We evaluate Eq.~\eqref{eq:topological-completion} as an interferometric probe on the fixed Einstein--Abelian-Higgs saddle. It therefore contributes no metric stress tensor and does not modify the vortex profiles used in Sec.~\ref{sec:background}. Promoting the mixed term to a fully dynamical Chern--Simons--Higgs theory would require re-solving the coupled matter equations and is not assumed here.

\subsection{Reflected-replica geometry and linking}
\label{sec:linking-geometry}

The smooth Euclidean nonrotating BTZ filling has the topology of a solid torus~\cite{BTZ1992}. The thermal Euclidean-time circle is the contractible meridian and closes smoothly at $r=r_+$, while the spatial angular circle is the noncontractible longitude. The time-reflection-symmetric exterior slice is an annulus bounded by the horizon and the conformal boundary. In particular, the coordinate region $r<r_+$, including the Lorentzian point $r=0$, is not part of the smooth Euclidean exterior.

The Lorentzian vortex fixes the quantised flux label $n$, but the exterior geometry does not select a unique embedding of a defect behind the horizon. In the reflected-replica path integral, the Poincar\'e-dual flux defect $\cV_n$ and its homology class are therefore part of the specified topological boundary condition. Denote
\begin{equation}
 \nu=\Lk(\Gamma_R,\cV_n).
 \label{eq:linking-class}
\end{equation}
For the connected RT saddle, the Dutta--Faulkner canonical purification glues the open EWCS $\Sigma_{A:B}$ to its mirror and produces a closed reflected contour $\Gamma_R$~\cite{DuttaFaulkner2021}. We analyse the linked replica sector in which this contour has prescribed integer linking $\nu_0$ with $\cV_n$. The exterior BTZ geometry alone does not force $\nu_0=1$. In the disconnected saddle the EWCS is absent, so the reflected probe contour associated with it is trivial. Figure~\ref{fig:reflected-linking} is a schematic projection of the unit-linking choice $\nu_0=1$, not a coordinate drawing of the Euclidean black-hole interior.

\begin{figure}[t]
 \centering
 \begin{tikzpicture}[>=Latex, scale=1.05]

 \begin{scope}[xshift=0cm]
  \draw[thick] (0,0) circle (2);
  \draw[thick, dashed] (0,0) circle (0.7);
  \fill[gray!20] (0,0) circle (0.7);
  \filldraw[red] (0,0) circle (0.06);
  \node[red,font=\footnotesize] at (0.25,-0.15) {$\cV_n$};
  \draw[very thick, blue] ({2*cos(30)},{2*sin(30)}) arc (30:110:2) node[midway,above] {$A$};
  \draw[very thick, blue] ({2*cos(-30)},{2*sin(-30)}) arc (-30:-110:2) node[midway,below] {$B$};
  \draw[thick, purple] ({2*cos(30)},{2*sin(30)}) .. controls (0.9,0.3) and (0.9,-0.3) .. ({2*cos(-30)},{2*sin(-30)});
  \draw[thick, purple] ({2*cos(110)},{2*sin(110)}) .. controls (-0.9,0.3) and (-0.9,-0.3) .. ({2*cos(-110)},{2*sin(-110)});
  \draw[very thick, orange] (-0.9,0) -- (0.9,0);
  \node[orange, font=\footnotesize] at (0.0,0.25) {$\Sigma_{A:B}$};
  \draw[thick, dotted, ->] (0.9,0) arc (0:80:0.9);
  \draw[thick, dotted] (0.9,0) arc (0:-80:0.9);
  \draw[thick, dotted] (-0.9,0) arc (180:100:0.9);
  \draw[thick, dotted] (-0.9,0) arc (180:260:0.9);
  \node[font=\footnotesize] at (1.3,0.55) {$\Gamma_R$};
  \node[below, font=\footnotesize] at (0,-2.2) {(a) connected linked sector, $\nu=1$};
 \end{scope}

 \begin{scope}[xshift=5cm]
  \draw[thick] (0,0) circle (2);
  \draw[thick, dashed] (0,0) circle (0.7);
  \fill[gray!20] (0,0) circle (0.7);
  \filldraw[red] (0,0) circle (0.06);
  \node[red,font=\footnotesize] at (0.25,-0.15) {$\cV_n$};
  \draw[very thick, blue] ({2*cos(30)},{2*sin(30)}) arc (30:110:2) node[midway,above] {$A$};
  \draw[very thick, blue] ({2*cos(-30)},{2*sin(-30)}) arc (-30:-110:2) node[midway,below] {$B$};
  \draw[thick, purple] ({2*cos(30)},{2*sin(30)}) .. controls (1.4,1.2) and (0.4,1.6) .. ({2*cos(110)},{2*sin(110)});
  \draw[thick, purple] ({2*cos(-30)},{2*sin(-30)}) .. controls (1.4,-1.2) and (0.4,-1.6) .. ({2*cos(-110)},{2*sin(-110)});
  \node[below, font=\footnotesize] at (0,-2.2) {(b) disconnected};
 \end{scope}
 \end{tikzpicture}
 \caption{Schematic reflected-replica linking sector. The drawing is a projection, not a coordinate representation of the region behind the Euclidean BTZ horizon. The exterior time-symmetric slice is the annulus between the gray horizon and the asymptotic boundary. The red point denotes the projected intersection of the specified topological flux defect $\cV_n$. (a) In the connected unit-linking sector, canonical-purification gluing of the EWCS $\Sigma_{A:B}$ (orange) to its mirror produces a closed reflected contour $\Gamma_R$ (dotted) with $\nu=1$. (b) In the disconnected saddle, the EWCS and its associated reflected probe contour are absent.}
 \label{fig:reflected-linking}
\end{figure}
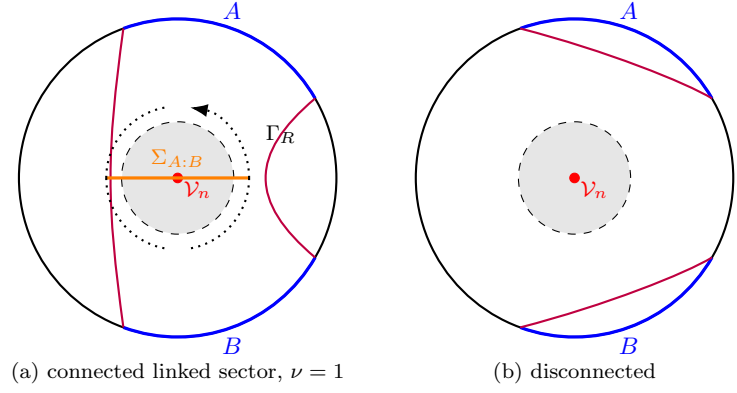

Denote the probe charge on $\Gamma_R$ by $p$. The normalised Wilson-threaded reflected moment is
\begin{equation}
 \Wtop_p(n)=
 \frac{
 Z^R[W_p(\Gamma_R)W_{\kappa n}(\cV_n)]\,Z^R[1]
 }{
 Z^R[W_p(\Gamma_R)]\,Z^R[W_{\kappa n}(\cV_n)]
 }.
 \label{eq:normalized-wilson-reflected}
\end{equation}
The normalisation is taken at fixed replica boundary state, fixed noncontractible holonomies, and identical framing. It removes the uncharged reflected saddle, local endpoint factors, and the self-linking phases of the individual lines, leaving the mutual braiding contribution in the local linked sector.

\subsection{The linking phase}

\begin{proposition}[Topological-hair phase]
\label{prop:linking}
For the probe theory~\eqref{eq:topological-completion}, evaluated in a fixed reflected-replica topological sector with $\nu=\Lk(\Gamma_R,\cV_n)$,
\begin{equation}
 \boxed{
 \Wtop_p(n)=
 \exp\!\left[
 \frac{2\pi i}{k}\,\kappa pn\nu
 \right].}
 \label{eq:topological-hair-phase}
\end{equation}
\end{proposition}

\begin{proof}
This is the standard Abelian Chern--Simons result of Witten and Polychronakos for two Wilson lines~\cite{Witten1989,Polychronakos1989}, with charges $p$ and $\kappa n$. A self-contained derivation, including the fixed-boundary-sector qualification required on a solid torus, is given in App.~\ref{app:CS-linking}. The normalisation in Eq.~\eqref{eq:normalized-wilson-reflected} cancels the one-line factors and leaves the mutual linking phase.
\end{proof}

Equation~\eqref{eq:topological-hair-phase} is independent of $G_N$, the vortex core radius, the Higgs and gauge masses, the black-hole temperature, and the interval lengths within a fixed topological sector. It changes only when the specified defect or probe contour changes linking class, or when the topological phase itself changes. Panel (b) of Fig.~\ref{fig:interferometer} shows the quantised phases for $U(1)_5$ with $\kappa=\nu=1$.

\subsection{Boundary current-algebra check}

The same phase follows from the boundary. The spectator $U(1)_k$ Chern--Simons theory induces a chiral $U(1)_k$ current algebra on the boundary. In the rational normalisation, a vertex operator $V_a$ has conformal weight $h_a=a^2/(2k)$ and the operator product
\begin{equation}
 V_p(z)V_{\kappa n}(0)=z^{\kappa pn/k}\big[V_{p+\kappa n}(0)+\cdots\big].
 \label{eq:boundary-vertex-OPE}
\end{equation}
Transporting $V_p$ once counterclockwise around $V_{\kappa n}$ sends $z\to e^{2\pi i}z$ and produces $M_{p,\kappa n}=e^{2\pi i\kappa pn/k}$. In the reflected replica, the permutation-twist part of the defect is neutral under this spectator current. Dividing by the correlators with either insertion alone cancels the ordinary twist correlator and the local vertex normalisations, so a contour with linking number $\nu$ contributes $M_{p,\kappa n}^{\nu}$, reproducing Eq.~\eqref{eq:topological-hair-phase}.

The winding sector is reconstructed by a discrete Fourier transform over the probe charge,
\begin{equation}
 \cP_a(n;\nu)=\frac1k\sum_{p=0}^{k-1}
 e^{-2\pi i pa/k}\Wtop_p(n)
 =\delta^{(k)}_{a,\kappa n\nu},
 \qquad
 \nu=\Lk(\Gamma_R,\cV_n).
 \label{eq:sector-reconstruction}
\end{equation}
For the minimal choice $\kappa=1$ and unit linking, this returns $n\bmod k$ exactly in the ideal topological sector. More generally it returns $\kappa n\nu\bmod k$ and determines $n$ whenever $\gcd(\kappa\nu,k)=1$; an unlinked contour returns the trivial sector. For a non-Abelian topological completion the normalised Hopf-link amplitude generalises to
\begin{equation}
 \Wtop_p(a)=\frac{S_{pa}S_{00}}{S_{p0}S_{0a}},
 \label{eq:nonabelian-link}
\end{equation}
with $S$ the modular $S$-matrix. Panel (c) of Fig.~\ref{fig:interferometer} shows the ideal sector reconstruction for $\kappa=\nu=1$.

\subsection{Replica-sector gating}
\label{sec:gating}

Combining the specified linking sector of Sec.~\ref{sec:linking-geometry} with the RT competition below gives the following conditional bridge between the two channels of Eq.~\eqref{eq:normalized-wilson-reflected}.

\begin{proposition}[Conditional replica-sector gating]
\label{prop:gating}
Fix a reflected-replica topological sector in which the closed contour associated with the connected EWCS has linking number $\nu_0$ with $\cV_n$. For two equal boundary intervals of length $\ell$ separated by $s$,
\begin{equation}
 \Lk(\Gamma_R,\cV_n)=
 \begin{cases}
 \nu_0,&\cC(\ell,s;\alpha)>0\quad\text{(connected)},\\
 0,&\cC(\ell,s;\alpha)<0\quad\text{(disconnected)},
 \end{cases}
 \label{eq:gating}
\end{equation}
where $\cC$ is the RT connectivity functional of Sec.~\ref{sec:geodesics}. The value $\nu_0$ is topological boundary data and is not fixed by the Euclidean BTZ exterior. For the unit-linking sector $\nu_0=1$, the RT transition $s=s_c(\ell,\alpha)$ switches the interferometric phase on or off. For the inverse-cube source, $\ell_\star r_+/L^2$ instead determines the sign of the first-order displacement of this gate.
\end{proposition}

\begin{proof}
When the connected saddle dominates, canonical-purification gluing produces the closed reflected contour, whose homology class is fixed to $\nu_0$ by the chosen replica sector. When the disconnected saddle dominates, the EWCS and its associated reflected probe contour are absent, so the normalised insertion is trivial. The sign reversal of $\dot s_c$ at Eq.~\eqref{eq:ell-star} is a separate geometric result proved in Sec.~\ref{sec:inverse-cube}.
\end{proof}

The falsifiable statement therefore has two parts: crossing $s_c(\ell,\alpha)$ switches the specified linked replica observable, while crossing $\ell_\star$ reverses the direction in which a positive vortex source moves that switching boundary.

\subsection{Phase--modulus factorisation and observable summary}

In the factorised spectator completion, the full Wilson-threaded moment has the schematic form
\begin{equation}
 Z^R_{m,r}(p;n)=
 Z^{R,\rm geom}_{m,r}[F_n]\,
 Z^{R,\rm end}_{m,r}(p)\,
 \exp\!\left[\frac{2\pi i}{k}\kappa pn\nu\right].
 \label{eq:phase-modulus-factorization}
\end{equation}
The first factor contains the vortex backreaction through the blackening function $F_n$. The second contains cutoff and charged-defect endpoint normalisation. The third is the universal topological hair. Equation~\eqref{eq:normalized-wilson-reflected} divides out the first two factors. Conversely, $|Z^R_{m,r}|$ and the uncharged Markov gap retain the geometric response. Table~\ref{tab:observable-separation} summarises the resulting split.

\begin{table}[t]
\centering
\caption{Information carried by the observables used in this work. ``Direct'' means that the quantity remains sensitive when the metric and local vortex profile are held fixed.}
\begin{ruledtabular}
\begin{tabular}{lccc}
Observable & Metric profile & Sector uncertainty & Winding sector \\
\hline
$\SR$ & yes & yes & not uniquely \\
$\DM=\SR-\MI$ & yes & cancels & no direct readout \\
standard $S_R(q)$ & yes & conditioned & equipartitioned at leading order \\
$\arg\Wtop_p(n)$ & no & no & direct for known $\kappa\nu$ \\
\end{tabular}
\end{ruledtabular}
\label{tab:observable-separation}
\end{table}

\section{Exact exterior geodesic problem}
\label{sec:geodesics}

On the time-reflection-symmetric slice of Eq.~\eqref{eq:metric},
\begin{equation}
 \dd s_\Sigma^2=\frac{\dd r^2}{F(r)}+\frac{r^2}{L^2}\dd x^2,
 \qquad r\geq r_+.
 \label{eq:static-slice}
\end{equation}
We work on the universal cover and choose two equal boundary intervals,
\begin{equation}
 A=[0,\ell],
 \qquad
 B=[\ell+s,2\ell+s].
 \label{eq:intervals}
\end{equation}

\subsection{RT geodesic of one interval}

For a curve $r(x)$, the length functional is $\cL=\int\dd x\sqrt{r'^2/F(r)+r^2/L^2}$. Translation invariance in $x$ gives, at the turning point $r=r_t$,
\begin{equation}
 \frac{\dd x}{\dd r}
 =\frac{Lr_t}{r\sqrt{F(r)}\sqrt{r^2-r_t^2}}.
 \label{eq:dxdr}
\end{equation}
An interval of width $w$ therefore satisfies
\begin{equation}
 \boxed{
 \frac{w}{2L}=\int_{r_t}^{\infty}
 \frac{r_t\,\dd r}
 {r\sqrt{F(r)}\sqrt{r^2-r_t^2}}
 =\int_0^1
 \frac{\dd u}{\sqrt{F(r_t/u)}\sqrt{1-u^2}}.}
 \label{eq:width-r}
\end{equation}
Because the integrand is real only where $F>0$, every finite-width turning point obeys $r_t(w)>r_+$: the RT geodesic never crosses the horizon on this static slice. The renormalised length is
\begin{equation}
 \cL_{\rm ren}(w)=2\int_0^1\!\left[
 \frac{r_t}{u^2\sqrt{F(r_t/u)}\sqrt{1-u^2}}
 -\frac{L}{u}\right]\dd u
 -2L\log\frac{2r_t}{L}.
 \label{eq:length-ren}
\end{equation}
The scheme constant cancels from mutual information.

\subsection{RT competition and EWCS}

For the intervals~\eqref{eq:intervals}, define
\begin{equation}
 \cC(\ell,s)=2\cL_{\rm ren}(\ell)
 -\cL_{\rm ren}(s)-\cL_{\rm ren}(2\ell+s).
 \label{eq:C-def}
\end{equation}
The connected RT saddle dominates when $\cC>0$; the transition is $\cC(\ell,s_c)=0$. Reflection symmetry fixes the EWCS to the radial segment between the turning points of the geodesics with widths $s$ and $2\ell+s$:
\begin{equation}
 \boxed{
 W(\ell,s)=\int_{r_t(2\ell+s)}^{r_t(s)}
 \frac{\dd r}{\sqrt{F(r)}}.}
 \label{eq:W-general}
\end{equation}
Both endpoints exceed $r_+$, so $W$ is entirely exterior. With $c_0=3L/(2G_N)$,
\begin{equation}
 \frac{\MI}{c_0}=\frac{\cC}{6L},\quad
 \frac{\SR}{c_0}=\frac{W}{3L},\quad
 \frac{\DM}{c_0}=\frac{2W-\cC}{6L}.
 \label{eq:info-general}
\end{equation}

\subsection{BTZ benchmark}

For $F=F_0$, Eq.~\eqref{eq:rt-BTZ} gives $r_t^{(0)}(w)=r_+\coth a_w$ with $a_w=r_+w/(2L^2)$. The thermal cross-ratio
\begin{equation}
 \zeta=\frac{\sinh^2a_\ell}
 {\sinh a_s\,\sinh a_{2\ell+s}}
 \label{eq:zeta}
\end{equation}
satisfies $\zeta>1$ in the connected phase and yields
\begin{align}
 \frac{\MI_0}{c_0}&=\frac13\log\zeta,
 \label{eq:I-BTZ}\\
 \frac{\SR^{(0)}}{c_0}&=\frac13\log
 \left(1+2\zeta+2\sqrt{\zeta(1+\zeta)}\right),
 \label{eq:SR-BTZ}\\
 \frac{\DM^{(0)}}{c_0}&=\frac13\log\left[
 \frac{1+2\zeta+2\sqrt{\zeta(1+\zeta)}}{\zeta}
 \right].
 \label{eq:Delta-BTZ}
\end{align}
The BTZ critical separation follows from $\zeta=1$: with $x=a_\ell$,
\begin{equation}
 \boxed{
 s_c^{(0)}=\frac{2L^2}{r_+}
 \left[
 \frac12\arcosh(2\cosh(2x)-1)-x
 \right].}
 \label{eq:sc-BTZ}
\end{equation}

\section{Response to an arbitrary weak vortex source}
\label{sec:linear}

Let $\mathcal P(r)=\alpha\pi(r)$ and $F(r;\alpha)=F_0(r)[1+\alpha h_\pi(r)]+O(\alpha^2)$, so that
\begin{equation}
 \boxed{
 h_\pi(r)=-\frac{1}{F_0(r)}
 \int_{r_+}^{r}\pi(u)F_0(u)\,\dd u.}
 \label{eq:source-to-h}
\end{equation}
A dot denotes $\partial_\alpha|_{\alpha=0}$, with the boundary widths held fixed.

\subsection{Turning-point response}

Differentiating Eq.~\eqref{eq:width-r} at fixed $w$ gives, with $\mathcal J_\pi(r_t)=\int_0^1 h_\pi(r_t/u)/[\sqrt{F_0(r_t/u)}\sqrt{1-u^2}]\,\dd u$,
\begin{equation}
 \boxed{
 \dot r_t(w)=-\frac{r_t^{(0)}(w)^2-r_+^2}{2L}
 \mathcal J_\pi(r_t^{(0)}(w)).}
 \label{eq:rt-response}
\end{equation}

\subsection{RT-length response}

With $\dot g_{rr}=-h_\pi(r)/F_0(r)$,
\begin{equation}
 \boxed{
 \dot\cL(w)=-\int_{r_t^{(0)}(w)}^\infty
 \frac{h_\pi(r)}{\sqrt{F_0(r)}}
 \sqrt{1-\frac{r_t^{(0)}(w)^2}{r^2}}\,\dd r,}
 \label{eq:length-response}
\end{equation}
and $\dot\cC=2\dot\cL(\ell)-\dot\cL(s)-\dot\cL(2\ell+s)$.

\subsection{EWCS response with moving endpoints}

Let $r_i=r_t^{(0)}(s)>r_o=r_t^{(0)}(2\ell+s)$. Differentiating Eq.~\eqref{eq:W-general} gives
\begin{equation}
 \boxed{
 \dot W=-\frac12\int_{r_o}^{r_i}
 \frac{h_\pi(r)}{\sqrt{F_0(r)}}\,\dd r
 +\frac{\dot r_i}{\sqrt{F_0(r_i)}}
 -\frac{\dot r_o}{\sqrt{F_0(r_o)}}.}
 \label{eq:W-response}
\end{equation}
The last two terms are required because the EWCS endpoints lie on RT surfaces whose turning radii move. The information responses are
\begin{equation}
 \boxed{
 \left.\partial_\alpha\frac{\DM}{c_0}\right|_0
 =\frac{2\dot W-\dot\cC}{6L}.}
 \label{eq:Delta-response}
\end{equation}
Equations~\eqref{eq:source-to-h}, \eqref{eq:length-response}, and~\eqref{eq:W-response} form a complete source-to-Markov-gap map for any weak numerical vortex profile.

\subsection{Transition shift}

Write $s_c(\alpha)=s_c^{(0)}+\alpha\dot s_c+O(\alpha^2)$. Differentiating $\cC(\ell,s_c;\alpha)=0$ and using $\partial_w\cL_0(w)=(r_+/L)\coth a_w$ gives
\begin{equation}
 \boxed{
 \dot s_c=
 \frac{\dot\cC(\ell,s_c^{(0)})}
 {(r_+/L)\left[
 \coth a_{s_c^{(0)}}+\coth a_{2\ell+s_c^{(0)}}
 \right]}.}
 \label{eq:sc-response}
\end{equation}

\section{Inverse-cube vortex tail}
\label{sec:inverse-cube}

For the source Eq.~\eqref{eq:inverse-cube-source}, $h_\pi=h_\star$ from Eq.~\eqref{eq:h-star}. With $q_w=r_+/r_t^{(0)}(w)=\tanh a_w$ and
\begin{equation}
 \mathfrak h(q,u)=h_\star(r_t/u)
 =\frac{q^2u^2}{2}
 +\frac{q^2u^2\log(qu)}{1-q^2u^2},
 \label{eq:h-qu}
\end{equation}
the one-dimensional response functions are (App.~\ref{app:kernels}):
\begin{align}
\frac{\dot\cL(w)}{L}&=-\int_0^1
 \frac{\mathfrak h(q_w,u)\sqrt{1-u^2}}
 {u\sqrt{1-q_w^2u^2}}\,\dd u,
 \label{eq:L-source}\\
 \frac{\dot W}{L}&=
 -\frac12\left[\mathfrak A(q_T)-\mathfrak A(q_s)\right]
 +\mathfrak E(q_s)-\mathfrak E(q_T),
 \label{eq:W-source}
\end{align}
where $\mathfrak A$ is the elementary primitive Eq.~\eqref{eq:A-primitive} and $\mathfrak E$ is the endpoint-motion contribution Eq.~\eqref{eq:E-endpoint}. Only smooth one-dimensional integrals remain.

\begin{figure*}[t]
 \centering
 \includegraphics[width=0.9\textwidth]{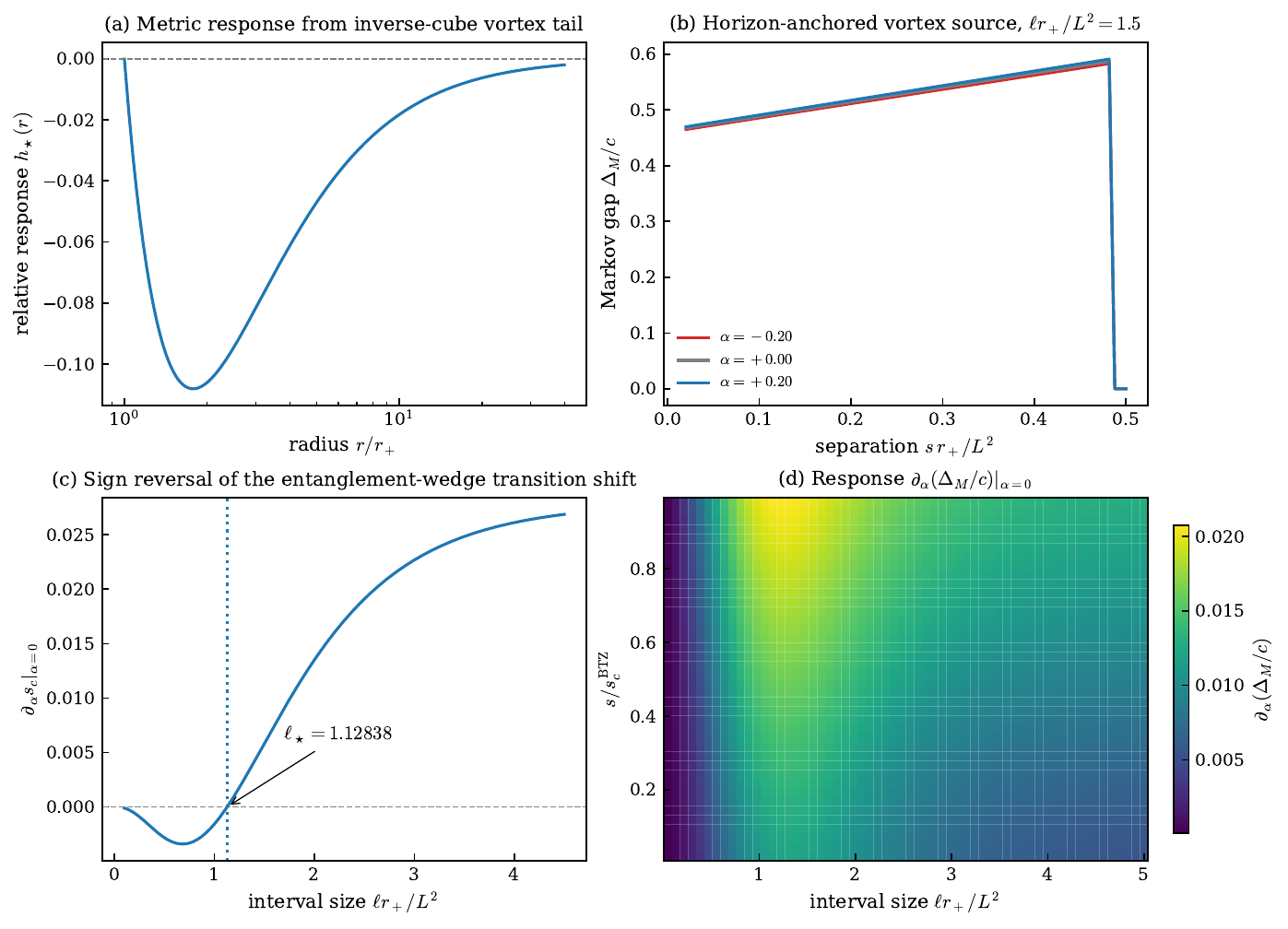}
 \caption{Complete geometric spine. (a) Relative metric response $h_\star=f_\star/F_0$ generated by the inverse-cube tail~\eqref{eq:inverse-cube-source}. It vanishes at the fixed horizon, is negative throughout the exterior, and approaches zero as $-\log(r/r_+)/r^2$ up to an overall constant. (b) Linearised Markov gap for $\ell r_+/L^2=1.5$ and three signs of $\alpha$. A positive horizon-anchored source increases the connected-phase gap; the sharp drop to zero marks the disconnected wedge. (c) First-order shift of the RT connectivity boundary as a function of interval size. The sign reverses at $\ell_\star r_+/L^2=1.128378\ldots$, computed to $10^{-10}$ from the exact metric~\eqref{eq:F-alpha-exact}. (d) Response $\partial_\alpha(\DM/c_0)|_{\alpha=0}$ across the connected wedge, positive throughout the sampled domain $0.05\le\ell r_+/L^2\le5$, $0.02\le s/s_c^{(0)}\le 0.98$.}
 \label{fig:geometric-spine}
\end{figure*}

Panels (b) and (d) of Fig.~\ref{fig:geometric-spine} confirm that
\begin{equation}
 \left.\partial_\alpha\frac{\DM}{c_0}\right|_0>0
 \label{eq:positive-response}
\end{equation}
throughout the sampled connected domain $0.05\le\ell r_+/L^2\le5$, $0.02\le s/s_c^{(0)}\le 0.98$. We verified the analytic derivative against centred finite differences of the exact metric Eq.~\eqref{eq:F-alpha-exact} to $10^{-6}$. Equation~\eqref{eq:positive-response} is a result for the inverse-cube representative, not a profile-independent positivity theorem. If the solution in winding sector $n$ has amplitude $\alpha_n$, then
\begin{equation}
 \DM^{(n)}=\DM^{(0)}
 +\alpha_n\left.\partial_\alpha\DM\right|_0
 +O(\alpha_n^2).
 \label{eq:Delta-n}
\end{equation}
The dependence on $n$ is mediated by $\alpha_n$ and, for the full solution, by changes in the shape of $\pi_n(r)$. Topology alone does not fix either.

\subsection{Sign reversal of the transition shift and interferometer gating}

Substitution of Eq.~\eqref{eq:L-source} into Eq.~\eqref{eq:sc-response} produces a dimensionless function of $\ell r_+/L^2$ with one nontrivial zero,
\begin{equation}
 \boxed{
 \frac{\ell_\star r_+}{L^2}=1.1283782224\ldots.}
 \label{eq:ell-star}
\end{equation}
For positive $\alpha$, $\dot s_c<0$ for $\ell<\ell_\star$ and $\dot s_c>0$ for $\ell>\ell_\star$: the same vortex tail destabilises the connected wedge for short intervals but stabilises it for long intervals. Panel (c) of Fig.~\ref{fig:geometric-spine} shows the sign reversal. Proposition~\ref{prop:gating} identifies the actual on/off gate as the RT boundary $s=s_c(\ell,\alpha)$; $\ell_\star$ determines only the direction in which the vortex shifts that gate. The sign reversal is a radial fingerprint rather than a topological invariant. Short geodesics emphasise the asymptotic tail, while long geodesics approach the horizon where $h_\star(r_+)=0$. The three lengths in $\cC$ combine those weights differently, allowing $\dot s_c$ to change sign even though $h_\star$ itself has a fixed sign.

\begin{figure*}[t]
 \centering
 \includegraphics[width=0.98\textwidth]{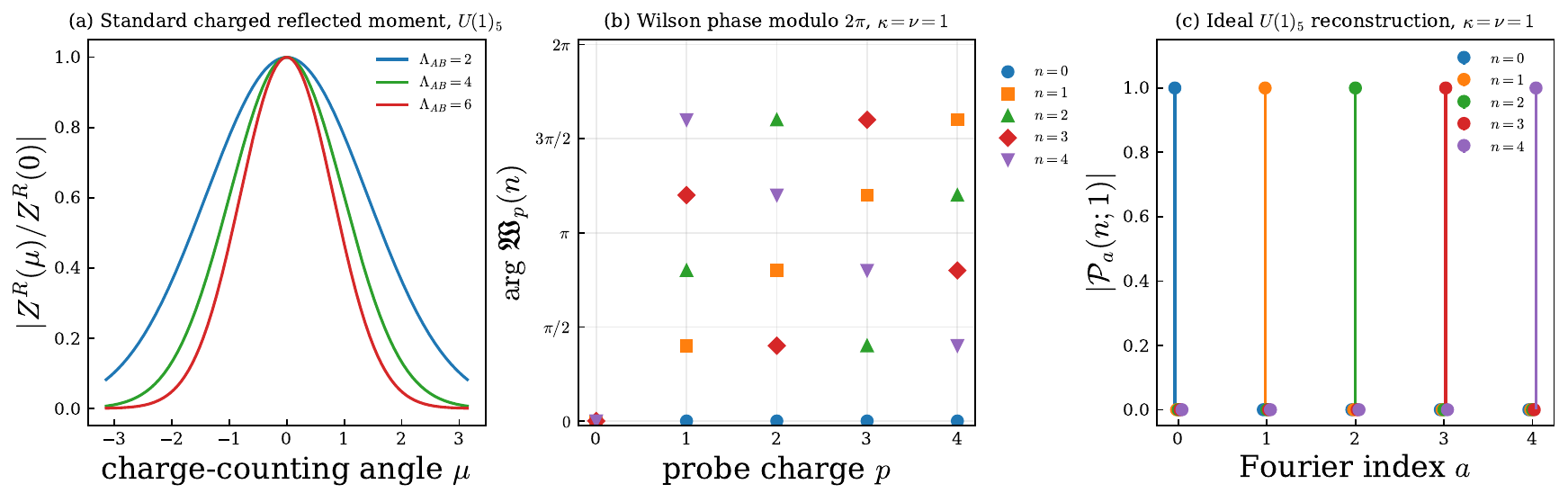}
 \caption{Topological-hair interferometer. (a) Modulus $|Z^R(\mu)/Z^R(0)|$ of the standard charged reflected moment Eq.~\eqref{eq:charged-reflected-gaussian} for a thermal $U(1)_5$ CFT at three interval scales. The modulus is Gaussian in the charge-counting angle $\mu$ with width set by $\Lambda_{AB}$; there is no winding-dependent phase in the standard charge-imbalance resolution. (b) Wilson-threaded phases $\arg\Wtop_p(n)$ modulo $2\pi$ for the ideal $U(1)_5$ probe sector with $\kappa=\nu=1$, $n=0,1,2,3,4$, and $p=0,\dots,4$. Smooth metric or core deformations cannot move these values while the topological sector is fixed. (c) Ideal discrete-Fourier reconstruction $\cP_a(n;1)$ of Eq.~\eqref{eq:sector-reconstruction}. For $\kappa=\nu=1$, the interferometer identifies $n\bmod k$ from $k-1$ nontrivial phase measurements. Standard symmetry-resolved reflected entropy, by contrast, is equipartitioned at leading order.}
 \label{fig:interferometer}
\end{figure*}

\section{Physical interpretation and boundary diagnostic}
\label{sec:interpretation}

\subsection{What the uncharged Markov gap measures}

The results separate three layers of information. The integer $n$ fixes the magnetic flux $2\pi n/e$. The coupled field equations map $n$ and the microscopic couplings to a radial source $\mathcal P_n(r)$. The gravitational constraint maps $\mathcal P_n(r)$ nonlocally to $F_n(r)$, after which the RT and EWCS kernels map $F_n$ to $\DM$:
\begin{equation}
 n\longrightarrow\mathcal P_n(r)
 \longrightarrow F_n(r)
 \longrightarrow\DM(A:B).
 \label{eq:causal-chain}
\end{equation}
This is not a universal topological relation: two vortex theories with the same $n$ but different core scales or couplings can produce different Markov gaps, and an ordinary matter profile can mimic the same geometric response. The superselection theorem gives an independent statement: a classical uncertainty over flux sectors contributes $H(p_n)$ to both $\SR$ and $\MI$, and hence not to $\DM$. For the Chern--Simons fixed-point states of Ref.~\cite{Berthiere2021}, the quantum-dimension terms also cancel. The uncharged Markov gap is therefore well suited to isolating geometric dressing from a decoupled topological sector, but it is not itself a winding meter.

\subsection{What ``flux-resolved'' can mean}

There are two inequivalent resolutions. The standard symmetry resolution of Sec.~\ref{sec:standard-resolution} conditions the reflected density matrix on the charge imbalance $q$ of the canonical purification. Its leading entropy is equipartitioned and $q$ is not the vortex winding. The topological resolution of Sec.~\ref{sec:wilson-interferometer} instead inserts a probe line and, for known $\kappa\nu$, uses its mutual-linking character to project onto the vortex sector $n\in\Zk$. The first construction is an entropy decomposition; the second is an interferometric sector measurement. Conflating them would lead to a spurious winding-dependent entropy. The positive universal quantity is the phase Eq.~\eqref{eq:topological-hair-phase}; the profile-dependent quantities are the modulus and the geometric Markov gap.

\subsection{Ensemble dependence}

Our calculation holds $r_+$ and the boundary interval coordinates fixed. The logarithmic term in Eq.~\eqref{eq:F-log-asymptotic} requires a subtraction scale, and varying $\alpha$ at fixed renormalised mass generally shifts $r_+$. The total response then contains
\begin{equation}
 \left.\frac{\dd\DM}{\dd\alpha}\right|_{M_{\rm ren}}
 =\left.\partial_\alpha\DM\right|_{r_+}
 +\left.\partial_{r_+}\DM\right|_\alpha
 \left.\frac{\dd r_+}{\dd\alpha}\right|_{M_{\rm ren}}.
 \label{eq:ensemble-chain}
\end{equation}
The first term is computed here. The second is an ensemble correction, not a new topological contribution.

\section{Condensed-matter bridge and measurement protocol}
\label{sec:measurement}

The separation developed above has a direct condensed-matter interpretation. The optimised Markov gap of Siva, Zou, Soejima, Mong, and Zaletel detects the minimal total central charge of an ungappable edge, $h_{\rm edge}=(c_+/3)\log2$, with numerical confirmation on Hofstadter Chern insulators achieving $|C|/3\,\log 2$ within $\sim 0.1\%$ using quadratic covariance-matrix disentanglers of radius $R$ optimised at each trisection~\cite{SivaEtAl2022}. This invariant diagnoses the topological phase but does not identify an Abelian anyon sector. Equation~\eqref{eq:topological-hair-phase} supplies the complementary information: the Markov gap reads edge content, while the Wilson-threaded phase reads the flux or anyon label. It is important not to identify the vortex winding $n$ with the Chern number $C$: the former labels an excitation sector, the latter characterises the bulk phase.

A minimal experimental protocol requires no full state tomography.
\begin{enumerate}
 \item Prepare an Abelian topological state and create a localised flux or anyon of sector $n$ inside the contour associated with the reflected partition. Programmable Rydberg arrays have already measured nonlocal string or loop operators in topologically ordered states~\cite{SemeghiniEtAl2021}.
 \item For each probe charge $p$, measure the generalised Wilson-loop amplitude in the vortex sector and in the reference sector without the flux defect. A controlled-loop Hadamard test gives
 \begin{equation}
  \mathcal A_p^{(n)}=\langle X_{\rm a}\rangle_n+i\langle Y_{\rm a}\rangle_n
  =\langle W_p(\Gamma_R)\rangle_n,
  \qquad
  \Wtop_p(n)=\frac{\mathcal A_p^{(n)}}{\mathcal A_p^{(0)}}.
  \label{eq:ancilla-readout}
 \end{equation}
 The calibrated ratio implements the normalisation in Eq.~\eqref{eq:normalized-wilson-reflected}. For a $\mathbb Z_2$ implementation, one nontrivial loop setting is sufficient and the phase is $(-1)^{\kappa pn\nu}$.
 \item Repeat for $p=0,\ldots,k-1$ and apply Eq.~\eqref{eq:sector-reconstruction}. This reconstructs $n\bmod k$ when $\kappa\nu$ is known and coprime to $k$. The protocol requires $k-1$ nontrivial probe settings, one ancillary degree of freedom when controlled loops are used, and circuit depth proportional to the loop length rather than the Hilbert-space dimension.
 \item Independently estimate $\MI$ and a reflected-entropy or CCNR proxy using two-copy interference or randomised measurements~\cite{RathEtAl2023,GoldsteinSela2018}. The RT gate is located by varying $s$ across $s_c(\ell,\alpha)$, while the geometric sign reversal is tested independently by varying $\ell$ across Eq.~\eqref{eq:ell-star}.
\end{enumerate}
The topological claim is falsified if the measured phase varies continuously under a deformation that neither crosses the vortex nor changes the topological phase. The geometric claim is falsified if the calibrated transition shift does not change sign at Eq.~\eqref{eq:ell-star} within the weak-source regime. The conditional gating statement (Proposition~\ref{prop:gating}) is falsified if, within a fixed linked replica sector, the reflected probe phase remains nontrivial after the EWCS disappears or remains trivial while the connected linked contour is present. These are three separate tests, which is an advantage: a failure of the microscopic vortex model does not masquerade as a failure of topological braiding, and vice versa. Cold-atom quantum-field simulators have already measured mutual information for spatially extended regions~\cite{TajikEtAl2023}, while Rydberg arrays have prepared and diagnosed topological spin-liquid states~\cite{SemeghiniEtAl2021}. The proposed interferometer combines those capabilities with charged-replica logic.

\section{Discussion and conclusions}
\label{sec:conclusion}

The result is a controlled separation theorem and a corresponding interferometer with a single geometric bridge. The uncharged Markov gap is blind to the Shannon uncertainty of superselection sectors (Proposition~\ref{prop:cancellation}); for Chern--Simons fixed-point states, the universal quantum-dimension contribution cancels as well. Standard $U(1)$ symmetry-resolved reflected entropy does not repair this limitation: the reflected charge imbalance is not the vortex winding, its expectation value in the canonical purification vanishes, and the adjacent-interval charged replica gives equipartition Eqs.~\eqref{eq:reflected-equipartition-renyi}--\eqref{eq:reflected-equipartition-vn}. A minimal probe $U(1)_k$ topological completion coupled to the compact vortex-flux class produces the observable $\Wtop_p(n)$ with the exact phase Eq.~\eqref{eq:topological-hair-phase} in a fixed linking sector. Its discrete Fourier transform reconstructs the winding sector modulo $k$ when $\gcd(\kappa\nu,k)=1$. Proposition~\ref{prop:gating} ties the on/off readout to the RT connectivity boundary $s=s_c(\ell,\alpha)$, while the independent scale $\ell_\star r_+/L^2=1.128378\ldots$ governs the sign of the vortex-induced displacement of that boundary.

The metric channel remains fully calculable. The radial Einstein equation must be solved with the integrating factor Eq.~\eqref{eq:F-exact-general}; at fixed horizon radius it gives $F'(r_+)=2r_+/L^2$. All same-boundary RT surfaces and the EWCS remain outside the horizon. The complete first variation includes the motion of the EWCS endpoints. For the inverse-cube source, the Markov gap increases over the connected domain sampled.

The central physical statement is
\begin{center}
 \fbox{\parbox{0.91\linewidth}{\centering
 The Wilson-threaded reflected phase reads topological hair in a specified linking sector, while the charged-moment modulus and the ordinary Markov gap read gravitational dressing. The RT transition $s=s_c(\ell,\alpha)$ is the gate; $\ell_\star r_+/L^2=1.128378\ldots$ fixes the direction of its vortex-induced shift.}}
\end{center}
This phase--modulus split is stronger than assigning a profile-dependent entropy correction to a vortex. It is universal in the topological sector, calculable in the gravitational sector, and directly translatable to an interferometric measurement. Recent multipartite Markov-gap and reflected-multi-entropy constructions suggest natural extensions to several reflected probe contours~\cite{IizukaEtAl2025,YuanLiZhou2024,BalasubramanianEtAl2025}; the bipartite result is complete without them. A next step is to solve the full coupled vortex profiles and to implement the same two-channel analysis in a fractional Chern or lattice gauge model. No universal winding-dependent entropy is assumed in that extension; the topological sector is read by braiding data, while entropy continues to diagnose geometry and edge structure.

\bibliographystyle{apsrev4-2}
\bibliography{ref_VBH}
\appendix
\onecolumngrid

\section{Gaussian Fourier transform and equipartition}
\label{app:equipartition}

Write the normalised charged reflected moment as $f_r(\mu)=\exp(-a_r\mu^2)$ with $a_r=k\Lambda_{AB}/(4\pi^2r)$. In the broad-distribution scaling regime,
\begin{align}
 P_r(q)&=\int_{-\infty}^{\infty}\frac{\dd\mu}{2\pi}
 e^{-iq\mu-a_r\mu^2}
 =\frac{1}{2\sqrt{\pi a_r}}
 \exp\!\left(-\frac{q^2}{4a_r}\right)\notag\\
 &=\sqrt{\frac{\pi r}{k\Lambda_{AB}}}
 \exp\!\left[-\frac{\pi^2rq^2}{k\Lambda_{AB}}\right].
\end{align}
The resolved R\'enyi entropy satisfies
\begin{equation}
 S_r^R(q)=S_r^R+\frac{1}{1-r}
 \left[\log P_r(q)-r\log P_1(q)\right].
\end{equation}
The terms proportional to $q^2$ cancel identically. The remainder is
\begin{equation}
 \log P_r-r\log P_1
 =\frac12\log r+\frac{1-r}{2}\log\frac{\pi}{k\Lambda_{AB}},
\end{equation}
which proves Eq.~\eqref{eq:reflected-equipartition-renyi}. Taking $r\to1$ and using $\lim_{r\to1}\log r/(1-r)=-1$ gives Eq.~\eqref{eq:reflected-equipartition-vn}. We verified these identities to machine precision.

\section{Chern--Simons derivation of the mutual-linking phase}
\label{app:CS-linking}

Equation~\eqref{eq:topological-hair-phase} is a special case of the standard Abelian $U(1)_k$ Chern--Simons partition function on a link~\cite{Witten1989,Polychronakos1989}. On a manifold with boundary, including the solid-torus BTZ filling, the absolute amplitude also depends on the chosen boundary state and on flat holonomies around noncontractible cycles. Equation~\eqref{eq:normalized-wilson-reflected} is therefore defined with identical boundary topological data and identical framing in all four partition functions. We further take the two linked insertions to lie in a contractible three-ball neighbourhood of the reflected filling, or equivalently condition on the same noncontractible holonomies. Under these restrictions the global solid-torus factors cancel and the remaining ratio is the local mutual-linking invariant.

For a collection of oriented Wilson lines $C_i$ with integer charges $q_i$, write the conserved two-form current $J_i$ such that $\int_{C_i}a=\int_{\mathcal M}a\wedge J_i$. The Euclidean action is
\begin{equation}
 I[a;J]=\frac{ik}{4\pi}\int a\wedge\dd a
 +i\sum_iq_i\int a\wedge J_i,
\end{equation}
with field equation $(k/2\pi)\dd a+\sum_iq_iJ_i=0$. Choose one-forms $\eta_i$ with $\dd\eta_i=J_i$ in the contractible neighbourhood specified above; this is the required Seifert-surface data. A saddle is $a_*=-(2\pi/k)\sum_iq_i\eta_i$. Substituting and evaluating gives, up to a convention-dependent orientation sign,
\begin{equation}
 \log Z[J]=\frac{2\pi i}{k}
 \sum_{i<j}q_iq_j\Lk(C_i,C_j)
 +\frac{\pi i}{k}\sum_iq_i^2\,\mathrm{SL}(C_i),
\end{equation}
where $\mathrm{SL}$ is the framing-dependent self-linking number. The ratio Eq.~\eqref{eq:normalized-wilson-reflected} cancels self-linking and single-line factors and leaves the mutual Gauss-linking number within the fixed boundary sector. Taking $C_1=\Gamma_R$, $q_1=p$, $C_2=\cV_n$, and $q_2=\kappa n$ yields Eq.~\eqref{eq:topological-hair-phase}. We verified numerically that the Gauss-linking integral for a standard Hopf link geometry gives $\Lk=\pm1$.

\section{Integrating-factor expansion}
\label{app:integrating-factor}

For $F'+\alpha\pi F=2r/L^2$ with $F(r_+)=0$, write $F=F_0+\alpha f+O(\alpha^2)$: the zeroth and first orders are $F_0'=2r/L^2$, $F_0(r_+)=0$ and $f'=-\pi F_0$, $f(r_+)=0$, proving Eq.~\eqref{eq:source-to-metric}. For $\pi=r_+^2/r^3$,
\begin{equation}
 f(r)=-\frac{r_+^2}{L^2}
 \left[
 \log\frac{r}{r_+}+\frac{r_+^2}{2r^2}-\frac12
 \right],
\end{equation}
which is Eq.~\eqref{eq:f-star}. For the exact source, Eq.~\eqref{eq:F-exact-general} contains
\begin{equation}
 \int u e^{-a/u^2}\dd u
 =\frac12\left[u^2e^{-a/u^2}+a\Ei(-a/u^2)\right],
\end{equation}
giving Eq.~\eqref{eq:F-alpha-exact}.

\section{Derivation of the geodesic equations}
\label{app:geodesics}

Let $\mathscr L=\sqrt{r'^2/F(r)+r^2/L^2}$. Because $\mathscr L$ has no explicit $x$-dependence, $\mathscr L-r'\partial_{r'}\mathscr L=r^2/(L^2\mathscr L)=r_t/L$ at $r'=0$. Solving for $r'$ yields Eq.~\eqref{eq:dxdr}. The on-shell line element obeys $\dd s/\dd r=r/[\sqrt{F(r)}\sqrt{r^2-r_t^2}]$, so the regulated length is $\cL(r_{\max})=2\int_{r_t}^{r_{\max}}r\,\dd r/[\sqrt{F(r)}\sqrt{r^2-r_t^2}]$. Changing variables to $u=r_t/r$ and subtracting the asymptotic $L/u$ divergence gives Eq.~\eqref{eq:length-ren}. For two equal intervals, the geometry is invariant under reflection about the midpoint. The minimal separating curve is radial and meets the two symmetric RT geodesics at their turning points, giving Eq.~\eqref{eq:W-general}.

\section{BTZ cross-ratio formulas}
\label{app:BTZ}

For $F_0=(r^2-r_+^2)/L^2$, Eq.~\eqref{eq:width-r} gives $w/(2L)=(L/r_+)\arctanh(r_+/r_t)$, so
\begin{equation}
 r_t^{(0)}(w)=r_+\coth a_w,\qquad a_w=\frac{r_+w}{2L^2}.
 \label{eq:rt-BTZ}
\end{equation} The finite interval length has the form $\cL_0(w)=2L\log\sinh a_w+\text{const}$, so $\cC_0/(2L)=\log\zeta$. The BTZ cross section is
\begin{align}
 W_0&=L\log\frac{
 r_t(s)+\sqrt{r_t(s)^2-r_+^2}}
 {r_t(2\ell+s)+\sqrt{r_t(2\ell+s)^2-r_+^2}},
\end{align}
and hyperbolic identities give Eq.~\eqref{eq:SR-BTZ}.

\section{Inverse-cube response kernels}
\label{app:kernels}

For $r=r_t/u$ and $q=r_+/r_t$, Eq.~\eqref{eq:h-star} becomes Eq.~\eqref{eq:h-qu}. Moreover, $\sqrt{F_0(r_t/u)}=(r_+/L)\sqrt{1-q^2u^2}/(qu)$. Substitution into the definition of $\mathcal J_\pi$ gives
\begin{equation}
 \mathcal J_\pi(r_t)=\frac{L}{r_+}\mathfrak J(q),
 \quad
 \mathfrak J(q)=\int_0^1
 \frac{qu\,\mathfrak h(q,u)}
 {\sqrt{1-u^2}\sqrt{1-q^2u^2}}\,\dd u,
\end{equation}
yielding
\begin{equation}
 \dot r_t(w)=-\frac{r_+(1-q_w^2)}{2q_w^2}
 \mathfrak J(q_w).
 \label{eq:rt-source-app}
\end{equation}
For the direct EWCS term, set $z=r_+/r$. Then
\begin{equation}
 \mathfrak A(z)=
 -\frac12\sqrt{1-z^2}
 +\frac{\log z}{\sqrt{1-z^2}}
 -\log\frac{z}{1+\sqrt{1-z^2}},
 \label{eq:A-primitive}
\end{equation}
with $\mathfrak A'(z)=h_\star(r_+/z)/[z\sqrt{1-z^2}]$, and the endpoint contribution is
\begin{equation}
 \mathfrak E(q)=-\frac{\sqrt{1-q^2}}{2q}\mathfrak J(q),
 \label{eq:E-endpoint}
\end{equation}
which is the finite version of $\dot r_t/\sqrt{F_0(r_t)}$. These give Eqs.~\eqref{eq:L-source} and~\eqref{eq:W-source}.

\end{document}